\documentclass[letterpaper, 10 pt, conference]{ieeeconf}  

\IEEEoverridecommandlockouts                              
\usepackage{graphics} 
\usepackage{graphicx}
\usepackage{booktabs} 
\usepackage{amsmath}
\usepackage{amssymb, amsthm}
\usepackage{caption, subcaption}
\usepackage{breqn}
\usepackage{array, mathtools}
\usepackage{pgfplots}
\usepackage{siunitx}
\usepackage{algorithm}
\usepackage[noend]{algpseudocode}
\usepackage{tabulary}
\usepackage{cite}

\usepackage{stmaryrd}
\usepackage{tikz}
\usepackage{pgfplots}
\usetikzlibrary{decorations.markings}
\usetikzlibrary{patterns}
\usetikzlibrary{arrows}
\usetikzlibrary{shapes}
\usetikzlibrary{fit}
 \usetikzlibrary{calc}
\usetikzlibrary{decorations.pathreplacing}
\usetikzlibrary{arrows,positioning} 
\usetikzlibrary{pgfplots.groupplots}
\pgfplotsset{compat=1.5}

\usepackage{thmtools}
\declaretheorem[style=definition,name=Definition,qed=$\blacksquare$]{definition}
\declaretheorem[style=definition,name=Example,qed=$\blacksquare$]{example}
\declaretheorem[style=definition,name=Special Case,qed=$\blacksquare$]{specialblock}
\declaretheorem[style=definition,name=Remark,qed=$\blacksquare$]{remark}

\declaretheorem[style=plain,name=Proposition]{proposition}
\declaretheorem[style=plain,name=Lemma]{lemma}
\declaretheorem[style=plain,name=Theorem]{theorem}
\declaretheorem[style=plain,name=Corollary]{corollary}

\newcommand{\R}{\mathbb{R}}
\newcommand{\eR}{\overline{\R}}
\renewcommand{\S}{\mathcal{S}}
\newcommand{\T}{\mathcal{T}}
\newcommand{\X}{\mathcal{X}}
\newcommand{\Y}{\mathcal{Y}}

\newcommand{\W}{\mathcal{W}}
\newcommand{\C}{\mathcal{C}}

\newcommand{\rect}[1]{\llbracket #1 \rrbracket}

\title{\LARGE {\bf Computing Robustly Forward Invariant Sets for Mixed-Monotone Systems}}

\author{Matthew Abate and Samuel Coogan
\thanks{This work was supported by the Air Force Office of Scientific Research under grant FA9550-19-1-0015 and by the National Science Foundation under grant $\#$1749357.}
\thanks{M. Abate is with the School of Mechanical Engineering and the School of Electrical and Computer Engineering, Georgia Institute of Technology, Atlanta, 30332, USA {\tt\small Matt.Abate@GaTech.edu}.}
\thanks{S. Coogan is with the School of Electrical and Computer Engineering and the School of Civil and Environmental Engineering, Georgia Institute of Technology, Atlanta, 30332, USA {\tt\small Sam.Coogan@GaTech.edu}.}
}

\begin{document}

\maketitle
\thispagestyle{empty}
\pagestyle{empty}

\begin{abstract}
This work presents new tools for studying reachability and set invariance for continuous-time mixed-monotone dynamical systems subject to a disturbance input. The vector field of a mixed-monotone system is decomposable via a decomposition function into increasing and decreasing components, and this decomposition enables embedding the original dynamics in a higher-dimensional embedding system. While the original system is subject to an unknown disturbance input, the embedding system has no disturbances and its trajectories provide bounds for finite-time reachable sets of the original dynamics.  Our main contribution is to show how one can efficiently identify robustly forward invariant and attractive sets for mixed-monotone systems by studying certain equilibria of this embedding system.  We show also how this approach, when applied to the backward-time dynamics, establishes different robustly forward invariant sets for the original dynamics. Lastly, we present an independent result for computing decomposition functions for  systems with polynomial dynamics. These tools and results are demonstrated through several examples and a case study.
\end{abstract}


\section{Introduction}
When verifying dynamical systems against safety constraints, it is often necessary to explicitly compute forward invariant subsets of the system state space.  Given a candidate subset, forward invariance can be shown by, e.g., studying the vector field on the boundary of the set \cite{blanchini1999set} or using barrier certificates \cite{prajna}; however, it is generally difficult to identify such candidates. In this paper, we provide several tools for identifying robustly forward invariant and attractive sets for continuous-time \emph{mixed-monotone} systems subject to a disturbance input.  A dynamical system is mixed-monotone if there exists a related \emph{decomposition function} that decomposes the system's vector field into increasing and decreasing components; mixed-monotonicity applies to continuous-time systems \cite{SuffMM, Coogan:2016, ANGELI20149, Nonmonotone, Gouze:1994qy}, discrete-time systems \cite{Smith2006}, as well as systems with disturbances \cite{LTLandMM, TIRA, coogan2015efficient}, and it generalizes the \emph{monotonicity} property of dynamical systems for which trajectories maintain a partial order over states \cite{smith2008monotone, monotonicity}.

In the case with no disturbance, it is known that a $2n$-dimensional symmetric \emph{embedding system} can be constructed from the decomposition function of an $n$-dimensional mixed-monotone system.  This embedding system is monotone with respect to a particular southeast order and the original dynamics are contained in an invariant $n$-dimensional \emph{diagonal} subspace.  Thus, tools from monotone systems theory can be applied to the embedding system to conclude properties of the original dynamics; in particular, such approaches are useful for stability analysis \cite{smith2008global, Chu1998}, reachability analysis \cite{Smith2006}, and formal verification and synthesis \cite{maxverifacation, maxsynth}. When disturbances are present, it is also possible to construct a monotone embedding system from the original dynamics. In this case, the embedding system is nondeterministic with a $2m$-dimensional disturbance input when the original system is subject to an $m$-dimensional disturbance input.  This result has been applied in discrete-time \cite{TIRA, coogan2015efficient} and in continuous-time \cite{LTLandMM, TIRA} for the computation of robust reachable sets.

In this work, we consider continuous-time mixed-monotone systems with disturbances, however, unlike \cite{LTLandMM, TIRA} we study a \emph{deterministic} embedding system that arises from considering the worst case disturbance inputs. 
While this deterministic embedding system is straightforwardly derived from the aforementioned nondeterministic embedding system, its potential does not seem to have been fully appreciated or studied in the literature. 
In particular, unlike the deterministic embedding system that arises in the case with no disturbance, the diagonal of this new deterministic embedding system is not forward invariant; instead, a forward invariant \emph{triangular} region is induced above the diagonal. 
Our main result is to show that equilibria in this triangular region correspond to robustly forward invariant sets for the original system and that stable equilibria correspond to attractive sets for the original system.

As a second contribution, we demonstrate a new approach for generating decomposition functions for systems with polynomial dynamics.
There do not exist universal algorithms for generating closed-form decomposition functions, except in a few, albeit important, special cases. In particular, it is observed in \cite{7799445} that a decomposition function can be constructed if each off-diagonal entry of the Jacobian matrix of the system's vector field does not change sign over the state space, and this result is extended in \cite{SuffMM, LTLandMM, TIRA} to system's with uniformly bounded Jacobian matrices. 
While this special case is quite general, the suggested construction can provide conservative approximations of, e.g., reachable sets, and we show through example that our proposed alternate decomposition function construction can be less conservative and applicable to systems not satisfying the special case described above.

As a third contribution, we show that the basic results discussed above for forward-time reachability analysis can be extended for backward-time reachability analysis in the same setting.
This result relies on the observation that if there exists a decomposition function for the backward-time dynamics, then approximation in the backward-time setting is possible using a method analogous to that used in the forward-time case.
Moreover, we show how the technique presented for obtaining rectangular forward invariant sets can be applied to the backward-time dynamics to obtain forward invariant sets for the original dynamics that are the complement of rectangular regions.

In summary, our main contributions are as follows: 
(a) we show that robustly forward invariant sets for continuous-time mixed-monotone systems with disturbances can be obtained by studying certain equilibria in an appropriate deterministic embedding system that differs from that studied in existing literature.  We show also how the attractivity of these sets can be determined by studying the stability of the equilibria. 
(b) We suggest a new procedure for computing decomposition functions for polynomial systems, and this method can be implemented in certain instances when others cannot.
(c) We present a method for over-approximating backward reachable sets for mixed-monotone systems and this method also enables identifying robustly forward invariant sets for the original dynamics.
The results and tools created in this work are demonstrated through three examples and a case study\footnote{The code that accompanies the examples and generates the figures in this work is publicly available through the GaTech Facts Lab GitHub: https://github.com/gtfactslab/Abate\_CDC2020\_2.}.

\section{Notation}
We denote the set of nonnegative and nonpositive real numbers by $\mathbb{R}_{\geq 0}$ and $\mathbb{R}_{\leq 0}$, respectively, and the extended real numbers by $\eR := \R\cup \{-\infty,\infty\}$, $\eR_{\geq 0} := \R_{\geq 0}\cup \{\infty\}$, and $\eR_{\leq 0} := \R_{\leq 0}\cup \{-\infty\}$.

Let $(x,\, y)$ denote the vector concatenation of $x,\, y \in \R^n$, i.e. $(x,\,y) := [x^T \, y^T]^T \in \R^{2n}$, and let $\preceq$ denote the componentwise vector order, i.e. $x\preceq y$ if and only if $x_i \leq y_i$ for all $i\in\{1,\cdots, n\}$ where vector components are indexed via subscript. Given $x, y\in\mathbb{R}^n$ with $x\preceq y$, 
\begin{equation*}
[x,\,y]:=\left\{z\in \R^n \,\mid\, x \preceq z \text{ and } z \preceq y\right \}
\end{equation*}
denotes the hyperrectangle defined by the endpoints $x$ and $y$, and
we extend this notation to componentwise inequality of matrices, i.e., $Z\in [X,\, Y]$ for $X,Y,Z\in\mathbb{R}^{n\times m}$ means each entry of $Z$ is lower and upper bounded by the entries of $X$ and $Y$, respectively.
We also allow $x\in \eR^n$ and $y\in\eR^n$, in which case $[x,\,y]$ defines an \emph{extended hyperrectangle}, that is, a hyperrectangle with possibly infinite extent in some coordinates.  
Given $a=(x,\,y) \in \R^{2n}$ with $x \preceq y$, we denote by $\rect{a}$ the hyperrectangle formed by the first and last $n$ components of $a$, i.e.,  $\rect{a}:=[x,\,y]$.

Let $\preceq_{\rm SE}$ denote the \emph{southeast order} on $\eR^{2n}$ defined by
\begin{equation*}
(x,\, x') \preceq_{\rm SE} (y,\, y')
 \:\: \Leftrightarrow  \:\:  x \preceq y\text{ and } y' \preceq x'
\end{equation*}
where $x,\, y,\, x',\, y' \in \eR^n$.  In the case that $x \preceq x'$ and $y \preceq y'$, observe that
\begin{equation}
\label{eq:order_to_box}
(x,\, x') \preceq_{\rm SE} (y,\, y')
 \:\: \Leftrightarrow  \:\:
[\,y,\, y'\,] \subseteq [\,x,\, x'\,].
\end{equation}

\section{Preliminaries on Mixed-Monotone Dynamical Systems}

Consider a dynamical system with disturbance input, i.e., a nondeterministic system, given by
\begin{equation}\label{eq1}
    \dot{x} = F(x,\, w)
\end{equation}
for Lipschitz $F$ where $x\in \X \subseteq \R^n$ and $w \in \W\subset \R^m$ denote the system state and a bounded time-varying disturbance, respectively.  We assume $\X$ is an extended hyperrectangle with nonempty interior and $\W$ is a hyperrectangle\footnote{The assumption that $\X$ is an extended hyperrectangle and $\W$ is a hyperrectangle can be relaxed for some of the results of this paper, but for ease of exposition, we make this assumption throughout.} so that $\W = [\underline{w},\, \overline{w}]$ for some $\underline{w},\, \overline{w}\in \R^m$ with $\underline{w} \preceq \overline{w}$.

For $T\geq 0$, let $\Phi^F(T; x_0,\mathbf{w})$ denote the (assumed unique) state of \eqref{eq1} reached at time $T$ starting from $x_0 \in \X$ at time $0$ under the piecewise continuous disturbance input $\mathbf{w}:[0,\, T]\to \W$.
We do not a priori require  $\Phi^F(T; x_0,\mathbf{w})$ to exist for all $T$; however, existence of $\Phi^F(T; x_0,\mathbf{w})$ implicitly means that $\Phi^F(t; x_0,\mathbf{w})\in \mathcal{X}$ for all $0\leq t\leq T$.
Additionally, let
\begin{multline}
\label{eq:reachable}
     R^{F}(T;\, \X_0) :=
     \Big{\{}\Phi^{F}(T;\, x_0,\, \mathbf{w}) \in \X \,\Big{|}\,  
     x_0\in \X_0
     \\ \text{for some } \mathbf{w} : [0,\, T] \rightarrow \W\Big{\}}
\end{multline}  
denote the set of states that are reachable by \eqref{eq1} in time $T\geq 0$ from $\X_0\subseteq \X$ under some disturbance input.

\begin{definition}\label{def1}
A set $A\subseteq \X$ is \emph{robustly forward invariant} for \eqref{eq1} if  $\Phi^F(T; x_0,\mathbf{w})\in A$ for all $x_0\in A$, all $T\geq 0$ and all piecewise continuous inputs $\mathbf{w}:[0,\, T]\to \W$ whenever $\Phi^F(T; x_0,\mathbf{w})$ exists. 
When $F$ does not depend on $w$ we simply say $A$ is \emph{forward invariant}.
\end{definition}

In this paper, we focus specifically on systems that are \emph{mixed-monotone} \cite{Coogan:2016}.

\begin{definition}\label{def2}
Given a locally Lipschitz continuous function $d : \X \times \W \times \X \times \W \rightarrow \mathbb{R}^n$, the system \eqref{eq1} is \emph{mixed-monotone with respect to $d$}  if all of the following hold:
\begin{itemize}
    \item For all $x \in \X$ and all $w \in \W$, $d(x,\, w,\, x,\, w) = F(x,\, w)$.
    \item For all $i,\, j \in \{1,\, \cdots,\, n\}$ with $i \neq j$,  $\frac{\partial d_i}{\partial x_j}(x,\, w,\,  \widehat{x},\,  \widehat{w}) \geq 0$ for all $x,\,  \widehat{x} \in \X$ and all $w,\,  \widehat{w} \in \W$ whenever the derivative exists.
    \item For all $i,\, j \in \{1,\, \cdots,\, n\}$, $\frac{\partial d_i}{\partial  \widehat{x}_j}(x, w,  \widehat{x}, \widehat{w}) \leq 0$ for all $x,\,  \widehat{x} \in \X$ and all $w,\,  \widehat{w} \in \W$ whenever the derivative exists.
    \item For all $i\in \{1,\, \cdots,\, n\}$ and all $k \in \{1,\, \cdots,\, m\}$, 
    $\frac{\partial d_i}{\partial w_k}(x,\, w,\,  \widehat{x},\,  \widehat{w}) \geq 0$ and $\frac{\partial d_i}{\partial  \widehat{w}_k}(x,\, w,\,  \widehat{x},\,  \widehat{w}) \leq 0$
    for  all $x,\,  \widehat{x} \in \X$ and all $w,\,  \widehat{w} \in \W$ whenever the derivative exists. \qedhere
\end{itemize}
\end{definition}

If \eqref{eq1} is mixed-monotone with respect to $d$, $d$ is said to be a \emph{decomposition function} for \eqref{eq1}, and when $d$ is clear from context we simply say \eqref{eq1} is mixed-monotone. 

There does not exist general algorithms for computing closed-form decomposition functions except for some albeit important special cases, such as those described below. As a rule of thumb, useful decomposition functions should be such that $d(x,w,\widehat{x},\widehat{w})$ is close to $F(x,w)$ when $x$ is close to $\widehat{x}$ and $w$ is close to $\widehat{w}$, but we do not provide a formal notion of closeness and observe that decomposition function construction usually leverages structural properties of $F$ or domain knowledge of the underlying physical system. 

We next present a special case for which the explicit construction of a decomposition function is possible. 
In particular, if each off-diagonal entry of $\frac{\partial F}{\partial x}$ and each entry of $\frac{\partial F}{\partial w}$ is either lower or upper bounded uniformly, then \eqref{eq1} is mixed-monotone and a decomposition function is constructed from $F$ and these bounds.

\begin{specialblock}\label{spec1}
  If there exists $\underline{J}_x\in \eR_{\leq 0}^{n\times n}$, $\overline{J}_x\in\eR_{\geq 0}^{n\times n}$, $\underline{J}_w\in\eR_{\leq 0}^{n\times m}$, and $\overline{J}_w\in\eR_{\geq 0}^{n\times m}$ such that 
\begin{itemize}
\item for all $x\in \X$ and all $w\in \W$, 
\begin{equation*}
\frac{\partial F}{\partial x}(x,w)\in [\underline{J}_x, \overline {J}_x] \:\text{ and }\: \frac{\partial F}{\partial w}(x,w)\in [\underline{J}_w, \overline {J}_w],
\end{equation*}
\item   for all $i\neq j$, 
$(\underline{J}_x)_{i,j}>-\infty \:\;\text{ or }\:\; (\overline{J}_x)_{i,j}<\infty,$
and
\item for all $i,\, k$, 
$(\underline{J}_w)_{i,k}>-\infty \:\; \text{ or }\:\; (\overline{J}_w)_{i,k}<\infty,$
\end{itemize}
then \eqref{eq1} is mixed-monotone and a decomposition function is constructed in the following way: 
\begin{enumerate}
\item  For all $i, j \in \{1,\cdots,n\}$ with $i \neq j$ and all $k \in \{1,\cdots, m\}$, choose  $\delta_{i,j},\,\epsilon_{i,k}\in\{0,1\}$ such that 
\begin{equation*}
\begin{array}{lcl}
    \delta_{i,j}=0 & \Rightarrow & (\underline{J}_x)_{i,j}\neq -\infty, \\
    \delta_{i,j}=1 & \Rightarrow & (\overline{J}_x)_{i,j}\neq \infty, \\
    \epsilon_{i,k}=0 & \Rightarrow & (\underline{J}_w)_{i,k}\neq -\infty, \\
    \epsilon_{i,k}=1 & \Rightarrow & (\overline{J}_w)_{i,k}\neq \infty.
\end{array}
\end{equation*}
Note that such a choice exists by hypothesis.
\item For all $i \in \{1,\cdots,n\}$, define $\xi^i,\, \alpha^i\in \mathbb{R}^n$ and $\pi^i,\, \beta^i\in \mathbb{R}^m$ element-wise according to
\begin{align*}
          (\xi^i_j,\alpha^i_j)&=
            \begin{cases}
                (x_i,0) & \text{if $i = j$,}\\
                (x_j,-(\underline{J}_x)_{i,j})&\text{if $i\neq j$ and $\delta_{i,j}=0$},\\
                (\widehat{x}_j,(\overline{J}_x)_{i,j})&\text{if $i\neq j$ and $\delta_{i,j}=1$},
            \end{cases}\\
            (\pi^i_k,\beta^i_k)&=
            \begin{cases}
                (w_k,-(\underline{J}_w)_{i,k})&\text{if $\epsilon_{i,k}=0$},\\
                (\widehat{w}_k,(\overline{J}_w)_{i,k})&\text{if $\epsilon_{i,k}=1$}.
            \end{cases}
\end{align*}
\item Define the $i^{\text{th}}$ element of $d$ according to
\begin{multline}
\label{eq:decomp}
d_i(x,w,\widehat{x},\widehat{w})=F_i(\xi^i,\pi^i) + (\alpha^i)^T(x-\widehat{x})    \\  + (\beta^i)^T(w-\widehat{w}),\:\:
\end{multline}
which is always well-defined on $\X\times \W \times \X \times \W$ since $\X$ and $\W$ are assumed to be hyperrectangles. \qedhere
\end{enumerate}
\end{specialblock}

\begin{remark}\label{rem1}
All \emph{monotone} dynamical systems satisfy the hypothesis of Special Case \ref{spec1}, and thus mixed-monotonicity generalizes the classical notion of monotonicity \cite{monotonicity}.  In particular, if \eqref{eq1} is monotone, i.e., 
\begin{itemize}
    \item for all $i\neq j$, $\frac{\partial F_i}{\partial x_j}(x,w)\geq 0$ for all $x\in \X$, $w\in \W$ whenever the derivative exists, and
    \item for all $i,k$, $\frac{\partial F_i}{\partial w_k}(x,w)\geq 0$ for all $x\in \X$, $w\in \W$ whenever the derivative exists,
\end{itemize}
then \eqref{eq1} is mixed-monotone with decomposition function $d(x,\, w,\, \widehat{x},\, \widehat{w})=F(x,\, w)$.
\end{remark}

A restrictive version of Special Case \ref{spec1} requiring sign-stability of the Jacobian matrices is first introduced in \cite{7799445}, and the essential observation that this extends to the case when the entries of the Jacobian are bounded is made in \cite{SuffMM} and is also in \cite{LTLandMM, TIRA}.  However, \cite{SuffMM, LTLandMM} require the diagonal entries of $\frac{\partial F}{\partial x}$ to be also bounded, and \cite{LTLandMM, TIRA} do not allow the other entries to be unbounded in one direction.  

The key feature of mixed-monotone systems that we exploit in this paper is that over-approximations of reachable sets can be efficiently computed by considering a deterministic auxiliary system constructed from the decomposition function.
We first consider the nondeterministic system
\begin{equation}\label{fakeembedding}
\begin{bmatrix}
  \dot{x}\\
  \dot{\widehat{x}} 
\end{bmatrix}
  = \varepsilon(x,\, w,\,  \widehat{x},\,\widehat{w})
  := 
\begin{bmatrix}
  d (x,\, w,\,  \widehat{x},\,\widehat{w})\\
  d ( \widehat{x},\,\widehat{w},\, x,\, w) 
\end{bmatrix}
\end{equation}
with state $(x,\, \widehat{x})\in \X\times\X$ and disturbance input $(w,\, \widehat{w})\in \W\times\W$.  We call \eqref{fakeembedding} the \emph{embedding system} relative to $d$, and we use $\Phi^{\varepsilon}( t;\, (\underline{x},\, \overline{x}), (\mathbf{w},\, \widehat{\mathbf{w}}))$ to denote the state of \eqref{fakeembedding} at time $t$ when initialized at $(\underline{x},\, \overline{x}) \in \X\times\X$ and when subjected to the piecewise continuous input $(\mathbf{w},\, \widehat{\mathbf{w}})$.  
Importantly, \eqref{fakeembedding} is a monotone control system as defined in \cite{monotonicity} when the orders on $\X\times \X$ and $\W\times \W$ are both taken to be the southeast orders; that is, if $a,\, a' \in \X\times \X$ and $\mathbf{b},\, \mathbf{b}': [0,\, \infty) \rightarrow \W\times\W$ satisfy $a \preceq_{\rm SE} a'$ and $\mathbf{b}(t) \preceq_{\rm SE} \mathbf{b}'(t)$ for all $t\geq 0$, then
\begin{equation}
    \Phi^{\varepsilon}( t;\, a,\, \mathbf{b}) \preceq_{\rm SE} \Phi^{\varepsilon}( t;\, a',\, \mathbf{b}')
\end{equation}
for all $t \geq 0$, provided $\Phi^{\varepsilon}(\, \cdot\,;\, a,\, \mathbf{b})$ and $\Phi^{\varepsilon}(\,\cdot\,;\, a',\, \mathbf{b}')$ remain in $\X\times \X$ on $[0,\,t]$.

Define $\Delta:=\{(x,\widehat{x})\in \X\times \X\mid x=\widehat{x}\}$ the \emph{diagonal} of the embedding system. 
Then for all $a\in \Delta$ and all $\mathbf{w}:[0,\infty)\to \W$ we have $\Phi^{\varepsilon}(t;\, a,\, (\mathbf{w},\mathbf{w}) )\in \Delta$ for all $t\geq 0$, i.e., $\Delta$ is robustly forward invariant for \eqref{fakeembedding} when the restriction $w=\widehat{w}$ is imposed.

Throughout most of this paper, we instead utilize a \emph{deterministic embedding system} given by
\begin{equation}\label{eq:embedding}
\begin{bmatrix}
  \dot{x}\\
  \dot{ \widehat{x}} 
\end{bmatrix}
  = e(x,\, \widehat{x})
  := 
\begin{bmatrix}
  d (x,\, \underline{w},\,  \widehat{x},\,\overline{w})\\
  d ( \widehat{x},\,\overline{w},\, x,\, \underline{w}) 
\end{bmatrix},
\end{equation}
with state transition function $\Phi^{e}( t;\, a) = \Phi^{\varepsilon}( t;\, a,\, (\underline{w},\, \overline{w}))$.  Note that
\begin{equation}
    \Phi^{e}( t;\, a) \preceq_{\rm SE} \Phi^{e}( t;\, a')
\end{equation}
for all $a,\, a' \in \X\times \X$ with $a \preceq_{\rm SE} a'$ and for all $t \geq 0$, i.e., \eqref{eq:embedding} is monotone with respect to the southeast order. However, unlike \eqref{fakeembedding}, $\Delta$ does not generally enjoy a forward invariance property for \eqref{eq:embedding} when $\underline{w}\neq \overline{w}$.

We next recall the following result establishing that the reachable set $R^F$ is over-approximated by solutions to the deterministic embedding system \eqref{eq:embedding}. The proof of this result appears in \cite[Appendix B1]{TIRA}, however, we provide our own proof here for completeness.

\begin{proposition}\label{prop:p1} 
Let \eqref{eq1} be mixed-monotone with respect to $d$, and consider $\X_0 = [\underline{x},\, \overline{x}]$ for some $\underline{x}\preceq \overline{x}$.
If $\Phi^{e}( t;\, (\underline{x},\, \overline{x}))\in \X\times \X$ for all $0\leq t\leq T$, then $R^F(T;\, \X_0) \subseteq \rect{\Phi^{e}( T;\, (\underline{x},\, \overline{x}))}.$
\end{proposition}
\begin{proof}
Choose $x \in [\underline{x},\, \overline{x}]$ and $\mathbf{w} : [0,\, T] \rightarrow \W$ for some $T \geq 0$.  Then from \eqref{fakeembedding} we have
\begin{equation*}
    (\Phi^F(T;\, x,\, \mathbf{w}), \Phi^F(T;\, x,\, \mathbf{w}))
    =
    \Phi^{\varepsilon}(T;\, (x,\, x),\, (\mathbf{w},\, \mathbf{w}) )
\end{equation*}
and
\begin{equation*}
    \Phi^{e}(T;\, (\underline{x},\, \overline{x}))
    =
    \Phi^{\varepsilon}(T;\, (\underline{x},\, \overline{x}),\, (\underline{w},\, \overline{w}) ).
\end{equation*}
Since $(\underline{x},\, \overline{x}) \preceq_{\rm SE} (x,\, x)$ and $(\underline{w},\, \overline{w})  \preceq_{\rm SE} (\mathbf{w}(t),\, \mathbf{w}(t))$ for all $t \in [0,\, T]$, we now have 
\begin{equation}
    \Phi^{e}(T;\, (\underline{x},\, \overline{x}))
    \preceq_{\rm SE}
    \begin{bmatrix}
    \Phi^F(T;\, x,\, \mathbf{w}) \\
    \Phi^F(T;\, x,\, \mathbf{w})
    \end{bmatrix}
\end{equation}
and thus $\Phi^F(T;\, x,\, \mathbf{w}) \in [\Phi^{e}(T;\, (\underline{x},\, \overline{x}))]$.  Therefore $R^F(T;\, \X_0) \subseteq [\Phi^{e}( T;\, (\underline{x},\, \overline{x}))]$.
\end{proof}

It is important to note that the usefulness of the mixed-monotonicity property for stability and reachability analysis---the main focus of this paper---is entirely dependent on the \emph{choice} of $d$. In general, a mixed-monotone system will be mixed-monotone with respect to many decomposition functions; however, certain decomposition functions may be more/less conservative than others when used with Proposition \ref{prop:p1}. One key observation of this paper is that the decomposition function construction presented as \eqref{eq:decomp} can be overly conservative or not possible, and we show through example how alternative decomposition functions are generally less conservative.

\section{Decomposition Functions for Polynomial Vector Fields}

Given the generality of the hypotheses of Special Case \ref{spec1}, two natural questions arise: First, are there systems that do not satisfy the hypotheses of Special Case \ref{spec1} but are nonetheless mixed-monotone with respect to some decomposition function other than \eqref{eq:decomp}? Second, for systems that do satisfy the hypotheses of Special Case \ref{spec1}, do there exist other, perhaps more useful, decomposition functions than \eqref{eq:decomp}? 
In the following example, we answer both questions affirmatively and illustrate a new technique for obtaining decomposition functions of polynomial systems.

\begin{figure*}[t!]
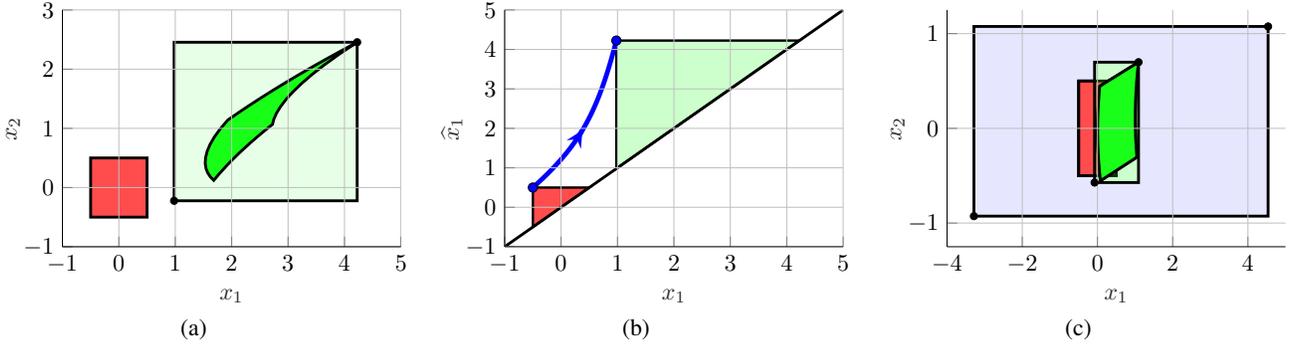

    \begin{subfigure}{0.31\textwidth}
    \scalebox{.9}{
%
%
\definecolor{mycolor1}{rgb}{0.92900,0.69400,0.12500}%

\begin{tikzpicture}
\path[use as bounding box](-1cm,-.8cm) rectangle +(6cm,4.4cm);
\begin{axis}[%
width=5cm,
height=3.5cm,
at={(0in,0in)},
scale only axis,
xmin=-1,
xmax=5,
xtick={-1,  0,  1,  2,  3,  4, 5}, 
xlabel style={font=\color{white!15!black}},
xlabel={$x_1$},
ymin=-1,
ymax=3,
ytick={-1, 0,  1, 2, 3},
ylabel style={at={(-0.1,0.5)}, font=\color{white!15!black}},
ylabel={$x_2$},
axis background/.style={fill=white},
axis x line*=bottom,
axis y line*=left,
xmajorgrids,
ymajorgrids,
axis on top
]

\addplot[area legend, line width=1.2pt, draw=black, fill=red, fill opacity=0.7, forget plot]
table[row sep=crcr] {%
x	y\\
-0.5	-0.5\\
-0.4	-0.5\\
-0.3	-0.5\\
-0.2	-0.5\\
-0.1	-0.5\\
0	-0.5\\
0.1	-0.5\\
0.2	-0.5\\
0.3	-0.5\\
0.4	-0.5\\
0.5	-0.5\\
0.5	-0.4\\
0.5	-0.3\\
0.5	-0.2\\
0.5	-0.1\\
0.5	0\\
0.5	0.1\\
0.5	0.2\\
0.5	0.3\\
0.5	0.4\\
0.5	0.5\\
0.4	0.5\\
0.3	0.5\\
0.2	0.5\\
0.1	0.5\\
0	0.5\\
-0.1	0.5\\
-0.2	0.5\\
-0.3	0.5\\
-0.4	0.5\\
-0.5	0.5\\
-0.5	0.4\\
-0.5	0.3\\
-0.5	0.2\\
-0.5	0.1\\
-0.5	0\\
-0.5	-0.1\\
-0.5	-0.2\\
-0.5	-0.3\\
-0.5	-0.4\\
-0.5	-0.5\\
}--cycle;

\addplot[area legend, line width=1.2pt, draw=black, fill=green, fill opacity=0.1, forget plot]
table[row sep=crcr] {%
x	y\\
0.979735229655718	-0.223702226579982\\
1.30423375319382	-0.223702226579982\\
1.62873227673191	-0.223702226579982\\
1.95323080027001	-0.223702226579982\\
2.27772932380811	-0.223702226579982\\
2.60222784734621	-0.223702226579982\\
2.92672637088431	-0.223702226579982\\
3.2512248944224	-0.223702226579982\\
3.5757234179605	-0.223702226579982\\
3.9002219414986	-0.223702226579982\\
4.2247204650367	-0.223702226579982\\
4.2247204650367	0.0441554684493733\\
4.2247204650367	0.312013163478729\\
4.2247204650367	0.579870858508084\\
4.2247204650367	0.847728553537439\\
4.2247204650367	1.11558624856679\\
4.2247204650367	1.38344394359615\\
4.2247204650367	1.6513016386255\\
4.2247204650367	1.91915933365486\\
4.2247204650367	2.18701702868422\\
4.2247204650367	2.45487472371357\\
3.9002219414986	2.45487472371357\\
3.5757234179605	2.45487472371357\\
3.2512248944224	2.45487472371357\\
2.92672637088431	2.45487472371357\\
2.60222784734621	2.45487472371357\\
2.27772932380811	2.45487472371357\\
1.95323080027001	2.45487472371357\\
1.62873227673191	2.45487472371357\\
1.30423375319382	2.45487472371357\\
0.979735229655718	2.45487472371357\\
0.979735229655718	2.18701702868422\\
0.979735229655718	1.91915933365486\\
0.979735229655718	1.6513016386255\\
0.979735229655718	1.38344394359615\\
0.979735229655718	1.11558624856679\\
0.979735229655718	0.847728553537439\\
0.979735229655718	0.579870858508084\\
0.979735229655718	0.312013163478729\\
0.979735229655718	0.0441554684493733\\
0.979735229655718	-0.223702226579982\\
}--cycle;

\addplot[area legend, line width=1.2pt, draw=black, fill=green, fill opacity=0.9, forget plot]
table[row sep=crcr] {%
x	y\\
1.68302138159568	0.121993086103313\\
1.76015774985729	0.209747379074114\\
1.8429050733731	0.298957361637399\\
1.93139932981606	0.389642913027619\\
2.02577962349615	0.481824229462162\\
2.12618826055129	0.575521829962475\\
2.23277082609746	0.670756562297627\\
2.34567626339393	0.767549609053254\\
2.4650569550808	0.865922493828845\\
2.59106880654812	0.965897087566484\\
2.72387133149772	1.06749561501416\\
2.74763306207465	1.1557763559593\\
2.79307395399336	1.25383784823433\\
2.86219108885708	1.36215532600483\\
2.95718437634688	1.48123629123835\\
3.08047619831661	1.61162285120856\\
3.23473316894653	1.75389425634589\\
3.42289025857682	1.90866965784314\\
3.64817756071541	2.07661110655015\\
3.91415001807577	2.25842681707799\\
4.2247204650367	2.45487472371357\\
3.94846751450065	2.31371978529659\\
3.68337188566274	2.17484265475353\\
3.42916022694129	2.038209473659\\
3.18556578792349	1.90378696895352\\
2.95232825032407	1.77154244129062\\
2.72919356364142	1.64144375364846\\
2.51591378536907	1.51345932019937\\
2.31224692562502	1.38755809543062\\
2.1179567960661	1.26370956351008\\
1.93281286295861	1.14188372789053\\
1.79691353306218	0.991056371990598\\
1.69217333313875	0.85257654190603\\
1.61573189596617	0.725818279977681\\
1.56499526278975	0.610196243191891\\
1.53761152462038	0.505162864037498\\
1.53144887813123	0.410205739530193\\
1.54457584037302	0.32484522807017\\
1.57524339571181	0.248632235802881\\
1.62186887401088	0.181146175940976\\
1.68302138159568	0.121993086103313\\
}--cycle;
\addplot[only marks, mark=*, mark options={}, mark size=1.5000pt, color=black, fill=black] table[row sep=crcr]{%
x	y\\
0.979735229655718	-0.223702226579982\\
};
\addplot[only marks, mark=*, mark options={}, mark size=1.5000pt, color=black, fill=black] table[row sep=crcr]{%
x	y\\
4.2247204650367	2.45487472371357\\
};

\end{axis}
\end{tikzpicture}%
        }
        \caption{ }
        \label{fig:ex1_1}
    \end{subfigure}
    ~
    \begin{subfigure}{0.31\textwidth}
    \scalebox{.9}{
        \input{Ex1_2.tikz}
        }
        \caption{ }
        \label{fig:ex1_2}
    \end{subfigure}
    ~
    \begin{subfigure}{0.31\textwidth}
    \scalebox{.9}{
        \input{Ex1_3.tikz}
        }
        \caption{ }
        \label{fig:ex1_3}
    \end{subfigure}
    \caption{  
    Approximating forward reachable sets of \eqref{eq3} from the set of initial conditions $\X_0 = [-1/2,\, 1/2]\times[-1/2,\, 1/2]$. (a) $\X_0$ is shown in red. $R^F(1;\, \X_0)$ is shown in green.
        The hyperrectangular over-approximation of $R^F(1;\, \X_0)$, that is computed from the embedding system \eqref{eq:embedding} as described in Proposition \ref{prop:p1}, is shown in light green. (b)  Visualisation of the bounding procedure from Proposition \ref{prop:p1}. 
        The trajectory of \eqref{eq:embedding} that yields Figure \ref{fig:ex1_1} is shown in blue, where $\Phi^e$ is projected to the $x_1,\, \widehat{x}_1$ plane. 
        The southeast cones corresponding to $\X_0$ and the hyperrectangular over-approximation of $R^F(1;\, \X_0)$ are shown in red and green, respectively. (c) Approximating $R^F(1/4;\, \X_0)$. $\X_0$ is shown in red. $R^F(1/4;\, \X_0)$ is shown in green with over-approximations derived from $d$ and $d'$ shown in green and blue, respectively.
    }
    \label{fig1}
\end{figure*}

\begin{example}\label{ex1}
Consider the system
\begin{equation}\label{eq3}
\begin{bmatrix}
\dot{x}_1 \\ \dot{x}_2
\end{bmatrix} = 
F(x) = 
\begin{bmatrix}
x_2^2 + 2 \\ x_1
\end{bmatrix}
\end{equation}
with $\X = \R^2$. Note that $\frac{\partial F_1}{\partial x_2} = 2x_2$ is neither lower or upper bounded on $\X$ and thus the system does not satisfy the hypotheses of Special Case \ref{spec1}.
However, \eqref{eq3} is mixed-monotone on $\X$ with decomposition function
\begin{equation}\label{eq4}
\begin{split}
     d_1(x,\, \widehat{x}) &= 
    \begin{cases}
    x_2^2 + 2
    & \text{if } x_2 \geq 0 \text{ and } x_2 \geq -\widehat{x}_2, \\
    \vspace{-.4cm} \\
    \widehat{x}_2^2 + 2
     & \text{if } \widehat{x}_2 \leq 0 \text{ and } x_2 < -\widehat{x}_2, \\ 
    \vspace{-.4cm} \\
    x_2 \widehat{x}_2 + 2
     & \text{if } x_2 < 0 \text{ and } \widehat{x}_2 > 0,
    \end{cases} \\
    d_2(x,\, \widehat{x}) &= x_1.
\end{split}
\end{equation}
Consider now a hyperrectangular set of initial conditions $\mathcal{X}_0$. Proposition \ref{prop:p1}  implies that the reachable set \eqref{eq:reachable} from $\mathcal{X}_0$ is approximated by a rectangular set defined from the state transition function of the $2n$ dimensional embedding system \eqref{eq:embedding}. An example is shown in Figures \ref{fig:ex1_1} and \ref{fig:ex1_2}.

Even though $\frac{\partial F}{\partial x}$ is not uniformly bounded on $\X = \R^2$, we can restrict our analysis to a compact subset $\X' \subset \X$ so that the decomposition function construction defined in Special Case \ref{spec1} is applicable. For instance, take $\X' = [-5,\, 5] \times [-5,\, 5]$, and note that $-10 \leq \frac{\partial F_1}{\partial x_2}(x) \leq 10$ for all $x \in \X'$. Applying Special Case \ref{spec1}, with $\delta_{1, 2} = \delta_{2, 1} = 1$, we have that
\begin{equation}
    d'(x,\, \widehat{x}) = 
    \begin{bmatrix}
    \widehat{x}^2_2 + 2 + 10(x_2 - \widehat{x}_2)
    \\
    x_1
    \end{bmatrix}
\end{equation}
is a decomposition function for \eqref{eq3} on $\X'$. Proposition \ref{prop:p1} then allows for computing reachable sets for \eqref{eq3} using $d'$ so long as the trajectories of the resulting embedding system remain within $\mathcal{X}'\times\mathcal{X}'$. An example is shown in Figure \ref{fig:ex1_3} where the reachable set computed using $d'$ is compared to the reachable set computed using $d$.  Note that, even though Special Case \ref{spec1} is made applicable by restricting the domain, the decomposition function given by \eqref{eq4} allows for a significantly tighter approximation of $R^F$.
\end{example}

Example \ref{ex1} suggests a new method for computing piecewise decomposition functions for \eqref{eq1} when $F$ is polynomial in $x$ and $w$.  This method has two steps:
\begin{enumerate}
\item Calculate all polynomial functions in $x,\, \widehat{x},\, w,\, \widehat{w}$ that evaluate to \eqref{eq1} when $x = \widehat{x}$ and $w = \widehat{w}$, and then
\item Form a continuous decomposition function as a piecewise combination of these polynomials, such that the remaining conditions from Definition \ref{def2} are satisfied.
\end{enumerate}
Due to space constraints, we do not present a formal algorithm for obtaining such decomposition functions, but the idea extends to systems with polynomial vector fields of arbitrary dimension and is applied in examples below.

\section{On Forward Invariance and Mixed-Monotone Systems}
In this section, we show how the embedding system \eqref{eq:embedding} can be used to efficiently compute sets that are robustly forward invariant for \eqref{eq1}.
Further, we leverage the monotonicity of \eqref{eq:embedding} to compute sets that are \emph{attractive} for \eqref{eq1}.

\begin{definition}[{\hspace{1sp}\cite{MAYHEW20111045}\hspace{1sp}}]\label{def3}
A set $A \subset \X$ is \emph{attractive from $\X' \subset \X$}, or simply \emph{attractive} for \eqref{eq1}, if for each solution $\Phi^F(\,\cdot\,;\, x_0,\, \mathbf{w})$ to \eqref{eq1} with $x_0\in \X'$ and piecewise continuous $\mathbf{w}$ and each relatively open neighborhood $\X_{\epsilon} \subset \X$ of $A$, there exists $T > 0$ such that $\Phi^F(t;x_0,\mathbf{w}) \in \X_{\epsilon}$ for all $t \geq T$.
When $\X'=\X$, we say $A$ is \emph{globally attractive}.
\end{definition}

While the nondeterministic embedding system \eqref{fakeembedding} has appeared in the literature before, along with connections to reachable set computations, the deterministic embedding system \eqref{eq:embedding} has not been fully considered. We begin with two lemmas on forward invariant regions for the embedding system \eqref{eq:embedding}.  These results are then related to invariant (and attractive) sets for \eqref{eq1} in Theorem \ref{thrm1}.

Define by $\T := \{(x,\, \widehat{x}) \in \X\times \X \,\vert\, x \preceq \widehat{x}\}$ the \emph{upper triangle} of the embedding system, and define by $\S := \{(x,\, \widehat{x}) \in \T \,\vert\, 0 \preceq_{\rm SE} e(x,\, \widehat{x}) \}$, the set of points in $\mathcal{T}$ such that the embedding system's vector field points into the southeast cone.

\begin{lemma}\label{lem1}
The set $\T$ is forward invariant for \eqref{eq:embedding}.
\end{lemma}
\begin{proof}
Choose $(\underline{x},\, \overline{x})\in \T$, then $\underline{x} \preceq \overline{x}$. Define 
\begin{equation}
    (x(t),\, \widehat{x}(t))
    = \Phi^e(t;\, (\underline{x},\, \overline{x})),
\end{equation}
where from \eqref{fakeembedding} we now have
\begin{align*}
    (x(t),\, \widehat{x}(t))
    & = \Phi^{\varepsilon}(t;\, (\underline{x},\, \overline{x}),\, (\underline{w},\, \overline{w})), \\
    (\widehat{x}(t),\, x(t))
    & = \Phi^{\varepsilon}(t;\, (\overline{x},\, \underline{x}),\, (\overline{w},\, \underline{w})).
\end{align*}
Since $(\underline{x},\, \overline{x}) \preceq_{\rm SE} (\overline{x},\, \underline{x})$ and $(\underline{w},\, \overline{w}) \preceq_{\rm SE} (\overline{w},\, \underline{w})$ we have $(x(t),\, \widehat{x}(t)) \preceq_{\rm SE} (\widehat{x}(t),\, x(t))$ for all $t \geq 0$.  Equivalently, $x(t) \preceq \widehat{x}(t)$ for all $t \geq 0$, and thus $\Phi^e(t;\, (\underline{x},\, \overline{x})) \in \T$.  Therefore, $\T$ is forward invariant for \eqref{eq:embedding}.
\end{proof}

\begin{lemma}\label{lem2}
The set $\S$ is forward invariant for \eqref{eq:embedding}, and $\Phi^e(t_1; a)\preceq_{\rm SE}\Phi^e(t_2; a)$ for all $a \in \S$ and all $0\leq t_1\leq t_2$.
\end{lemma}

Lemma \ref{lem2} is a direct result of \cite[Ch. 3, Prop. 2.1]{smith2008monotone}.

We next present our main result and show how forward invariant and attractive regions can be identified via stability analysis for the embedding system \eqref{eq:embedding}.

\begin{theorem}\label{thrm1}
Suppose \eqref{eq1} is mixed-monotone with respect to $d$.  Then for all $a \in \S$ the following hold:
\begin{enumerate}
    \item For all $T \geq 0$, the set $\rect{\Phi^{e}(T;\, a)} \subseteq \X$ is robustly forward invariant for \eqref{eq1}.
    \item The limit $\lim_{t \rightarrow \infty}\Phi^{e}(t;\, a) =: (x_{\rm eq},\widehat{x}_{\rm eq})$ exists and $e(x_{\rm eq},\widehat{x}_{\rm eq}) = 0$, i.e., $(x_{\rm eq},\widehat{x}_{\rm eq})$ is an equilibrium for the deterministic embedding system \eqref{eq:embedding}.
    \item The set $[x_{\rm eq},\, \widehat{x}_{\rm eq}]$ is robustly  forward invariant and attractive from $\rect{a} \subseteq \X$. \qedhere
\end{enumerate}
\end{theorem}
\begin{proof}
Part 1. Suppose $\S$ is nonempty, and choose $a \in \S$. Then, from Lemma \ref{lem2}, $\Phi^e(t;\, a) \in \S$ for all $t \geq 0$. Choose $T\geq 0$ and let $b=\Phi^e(T; a)\in \S$. Also from Lemma \ref{lem2}, $b\preceq_{\rm SE} \Phi^e(t; b)$ so that $\rect{\Phi^{e}(t;\, b)} \subseteq \rect{b}$ by \eqref{eq:order_to_box} for all $t\geq 0$. 
From Proposition \ref{prop:p1} we have $R^F(t;\, \rect{b}) \subseteq \rect{\Phi^{e}(t;\, b)}$.  Therefore $R^F(t;\, \rect{b}) \subseteq \rect{b}$ for all $t\geq 0$, i.e., 
$\rect{b}$ is robustly forward invariant for \eqref{eq1}.  This completes the proof of the first part since $T\geq 0$ was arbitrary.

Part 2. This result follows from \cite[Ch. 3, Prop. 2.1]{smith2008monotone} applied to the monotone embedding system. In particular, since $a \preceq_{\rm SE} \Phi^{e}(t;\, a)$ for all $t \geq 0$, and $\T$ is forward invariant for \eqref{eq:embedding},  we have
    $\Phi^e(t;\, a) \in \{ (\underline{b},\, \overline{b})\in \X\times \X \;\mid\; \underline{a} \preceq \underline{b} \preceq \overline{b} \preceq \overline{a} \}$
for all $t \geq 0$, where we define $\underline{a},\, \overline{a} \in \X$ by $a = (\underline{a},\, \overline{a})$. Since $\Phi^e(t;\, a)$ is increasing with respect to the southeast order and is bounded,  $\lim_{t \rightarrow \infty}\Phi^{e}(t;\, a) := (x_{\rm eq},\, \widehat{x}_{\rm eq})$ exists and $e(x_{\rm eq},\widehat{x}_{\rm eq}) = 0$.

Part 3. Choose $x \in \rect{a}$ and $\mathbf{w}: [0,\, \infty]\rightarrow \W$. Then
\begin{equation*}
    (\Phi^F(t;\, x,\, \mathbf{w}), \Phi^F(t;\, x,\, \mathbf{w}))
    =
    \Phi^{\varepsilon}(t;\, (x,\, x),\, (\mathbf{w},\, \mathbf{w}) )
\end{equation*}
and
\begin{equation*}
    \Phi^{e}(t;\, a)
    =
    \Phi^{\varepsilon}(t;\, a,\, (\underline{w},\, \overline{w}) )
\end{equation*}
hold for all $t \geq 0$. Since $a \preceq_{\rm SE} (x,\, x)$ and $(\underline{w},\, \overline{w}) ) \preceq_{\rm SE} (\mathbf{w}(t),\, \mathbf{w}(t))$ for all $t \geq 0$, we now have $\Phi^F(t;\, x,\, \mathbf{w}) \in \rect{\Phi^{e}(t;\, a)}$ for all $t \geq 0$.  Choose a relatively open neighborhood $\X_{\epsilon}$ of $[x_{\rm eq},\, \widehat{x}_{\rm eq}]$ and a relatively open ball $B \subset \X\times \X$ such that $(x_{\rm eq},\, \widehat{x}_{\rm eq}) \in B \subset \X_{\epsilon} \times \X_{\epsilon}$.  From Part 2, there must exist a $T \geq 0$ such that $\Phi^{e}(T;\, a) \in B$ and at this time $\Phi^F(T;\, x,\, \mathbf{w}) \in \X_{\epsilon}$.
From Part 1 we have that $\rect{\Phi^{e}(T;\, a)}$ is robustly forward invariant for \eqref{eq1} and therefore $\Phi^F(t;\, x,\, \mathbf{w}) \in \X_{\epsilon}$ for all $t \geq T$.
Therefore, $[x_{\rm eq},\, \widehat{x}_{\rm eq}]$ is attractive on \eqref{eq1} from $\rect{a}$. The fact that $[x_{\rm eq},\, \widehat{x}_{\rm eq}]$ is robustly forward invariant follows immediately from Part 1.  
\end{proof}

Theorem \ref{thrm1} provides a basic algorithm for identifying invariant sets; if \eqref{eq:embedding} has an \emph{equilibrium}, i.e. if there exists an $(x_{\rm eq},\, \widehat{x}_{\rm eq})\in \T$ such that $e(x_{\rm eq},\, \widehat{x}_{\rm eq})=0,$ then $[x_{\rm eq},\, \widehat{x}_{\rm eq}]$ is robustly forward invariant for \eqref{eq1}.
Computing equilibria for \eqref{eq:embedding} requires solving a system of $2n$ nonlinear equations and, therefore, is generally computationally tractable.
Moreover, if a point $a \in \S$ is known, then one can simulate the embedding dynamics forward in time, starting from $a$, in order to find an equilibria; see Theorem \ref{thrm1} Part 2.

In the following two corollaries, we show how globally attractive regions for \eqref{eq1} can be identified via stability analysis in the embedding space.

\begin{corollary}\label{cor}
Suppose \eqref{eq1} is mixed-monotone with respect to $d$. If $(x_{\rm eq},\, \widehat{x}_{\rm eq}) \in \T$ is an asymptotically stable equilibrium for \eqref{eq:embedding} with a basin of attraction $\C \subseteq \X\times\X$, then $[x_{\rm eq},\, \widehat{x}_{\rm eq}]$ is robustly forward invariant for \eqref{eq1} and attractive from all $\rect{a}$ such that $a \in \C \cap \T$. In particular, if $\T \subseteq \C$, then $[x_{\rm eq},\, \widehat{x}_{\rm eq}]$ is globally attractive and robustly forward invariant for \eqref{eq1}.
\end{corollary}
\begin{proof}
Robust forward invariance of $[x_{\rm eq},\, \widehat{x}_{\rm eq}]$ follows immediately from Theorem \ref{thrm1} Part 1. Attractivity of $[x_{\rm eq},\, \widehat{x}_{\rm eq}]$ follows by a slight modification of the proof of Theorem \ref{thrm1} Part 3, where we observe that Part 2 of the theorem is invoked to establish that $\Phi^e(T; a)\in B$ for some $T\geq 0$, but this now holds
for all $a \in \C$.  Thus $\lim_{t \rightarrow \infty}\Phi^{e}(t;\, a) = (x_{\rm eq},\, \widehat{x}_{\rm eq})\in B$.  
\qedhere
\end{proof}

It is instructive to consider the specialization of Theorem \ref{thrm1} to monotone systems.
\begin{corollary}\label{cor3}
Suppose \eqref{eq1} is monotone, i.e., satisfies the conditions of Remark \ref{rem1}. If $x_{\rm eq}\in \X$ is globally asymptotically stable for $\dot{x} = F(x,\, \underline{w})$ and $\widehat{x}_{\rm eq} \in \X$ is globally asymptotically stable for $\dot{x} = F(x,\, \overline{w})$, then $\X_{\rm eq} = [x_{\rm eq},\, \widehat{x}_{\rm eq}]$ is robustly forward invariant and globally attractive for \eqref{eq1}.  Additionally, no hyperrectangle that is a proper subset of $\X_{\rm eq}$ is robustly forward invariant for \eqref{eq1}.
\end{corollary}
\begin{proof}
If \eqref{eq1} is monotone then $d(x,\, w,\, \widehat{x},\, \widehat{w}) = F(x,\, w)$ is a decomposition function for \eqref{eq1}; see Remark \ref{rem1}. 
Thus if $x_{\rm eq}\in \X$ is globally asymptotically stable for $\dot{x} = F(x,\, \underline{w})$ and $\widehat{x}_{\rm eq} \in \X$ is 
globally asymptotically stable for $\dot{x} = F(x,\, \underline{w})$, then
$(x_{\rm eq},\, \widehat{x}_{\rm eq})$ globally asymptotically stable for \eqref{eq:embedding}, and from Corollary \ref{cor} we have that $\X_{\rm eq}$ is globally attractive on \eqref{eq1}.  
Moreover, no proper hyper-rectangular subset of $\X_{\rm eq}$ can be robustly forward invariant on \eqref{eq1} as there exist trajectories of \eqref{eq1} that begin in $\X_{\rm eq}$ and reach $x_{\rm eq}$ (and $\widehat{x}_{\rm eq}$). \qedhere
\end{proof}

We demonstrate the applicability of Theorem \ref{thrm1} for computing forward invariant regions in the following example.

\begin{example}\label{Example2}
Consider the system
\begin{equation} \label{eq8}
\begin{bmatrix}
\dot{x}_1 \\ \dot{x}_2
\end{bmatrix}
=
F(x,\, w)
=
\begin{bmatrix}
-x_1 - x_1^3 - x_2 - w\\ -x_2 - x_2^3 + x_1 +w^3
\end{bmatrix}
\end{equation}
with $\X = \R^2$ and $\W = [-2,\, 2]$.  The system \eqref{eq8} is mixed-monotone with decomposition function
\begin{equation}
d(x,\, w,\, \widehat{x},\, \widehat{w})
=
\begin{bmatrix}
-x_1 - x_1^3 - \widehat{x}_2 - \widehat{w}\\ -x_2 - x_2^3 + x_1 +w^3
\end{bmatrix}.
\end{equation}
Additionally,
$ e(x_{\rm eq},\, \widehat{x}_{\rm eq}) = 0$
for 
\begin{equation}\label{eq25}
x_{\rm eq} = (-1.37,\, -1.95), \qquad
\widehat{x}_{\rm eq} = (1.37,\, 1.95),
\end{equation}
and $(x_{\rm eq},\,\widehat{x}_{\rm eq})\in \T$.  Therefore, from Theorem \ref{thrm1}, $\X_{\rm eq} = [x_{\rm eq},\, \widehat{x}_{\rm eq}]$ is robustly forward invariant for \eqref{eq8}.
Additionally, it can be checked that $(x_{\rm eq},\, \widehat{x}_{\rm eq})$ is globally asymptotically stable for \eqref{eq:embedding}; evoking Corollary \ref{cor}, we now have that $\X_{\rm eq}$ is globally attractive for \eqref{eq8}. We show $\X_{\rm eq}$ graphically in Figure \ref{fig2}.
\end{example}

\begin{figure}
    \hspace{.1in}
    \input{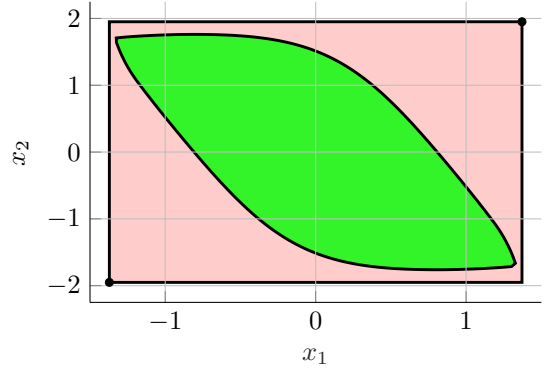}
    \caption{Computing attractive sets of \eqref{eq8} on $\X = \R^2$.
    $\X_{\rm eq}$ from \eqref{eq25} is shown in red.
    The region shown in green is globally attractive for \eqref{eq8}, and no proper subset of this region is robustly forward invariant.
    }
    \label{fig2}
\end{figure}

\section{Backward-Time Reachability for Mixed-Monotone Systems and Invariance}

In this section, we present a result analogous to Proposition \ref{prop:p1} for over-approximating finite-time backward reachable sets.  Later in the section, we leverage this result for the computation of robustly forward invariant regions for \eqref{eq1}.

The system \eqref{eq1} induces the backward-time dynamics
\begin{equation}\label{eq5}
\dot{x} = G(x,\, w) := - F(x,\, w)
\end{equation}
with $x \in \X$ and $w \in \W$, and \eqref{eq1} and \eqref{eq5} are related in the following way: if $x_1 = \Phi^F(T;\, x_0,\, \mathbf{w})$ for $\mathbf{w}:[0,\, T] \rightarrow \W$, then $x_0 = \Phi^G(T;\, x_1,\, \mathbf{w}')$ for $\mathbf{w}'(t) = \mathbf{w}(T - t)$, where $\Phi^G$ denotes the state transition function of \eqref{eq5}.
Let
\begin{multline}\label{micheal}
    S^{F}(T;\, \X_1) :=
    \Big{\{}
	x_0 \in \X
	\,\Big{|}\,
     \Phi^{F}(T;\, x_0,\, \mathbf{w}) \in \X_1\\
        \text{ for some }\mathbf{w} : [0,\, T] \rightarrow \W\Big{\}}
\end{multline}  
denote the set of initial conditions for which there exists a $\mathbf{w} :[0,\, T] \rightarrow \W$ capable of driving \eqref{eq1} to the set $\X_1$ in time $T\geq 0$.  Note that $R^G(T;\, \X_1) = S^F(T; \X_1)$ where $R^G(T;\, \X_1)$ is given by \eqref{eq:reachable}.
We next show that if \eqref{eq5} is mixed-monotone, then $S^F$ can be approximated using a procedure similar to that presented in Proposition \ref{prop:p1}.

\begin{proposition}\label{prop:p3}
Let \eqref{eq5} be mixed-monotone with respect to $D$, and choose $\X_1 = [\underline{x},\, \overline{x}]$. 
Construct the deterministic embedding system
\begin{equation}
\begin{bmatrix}
  \dot{x}\\
  \dot{ \widehat{x}} 
\end{bmatrix}
  = E(x,\, \widehat{x})
  := 
\begin{bmatrix}
  D(x,\, \underline{w},\,  \widehat{x},\,\overline{w})\\
  D( \widehat{x},\,\overline{w},\, x,\, \underline{w}) 
\end{bmatrix}
\end{equation}
with state transition function $\Phi^E$. 
If $\Phi^{E}( t;\, (\underline{x}, \overline{x}))\in \X\times \X$ for all $0\leq t\leq T$, then $S^F(T;\, \X_1) \subseteq [\Phi^{E}( T;\, (\underline{x},\, \overline{x})) ].$
\end{proposition}
\begin{proof}
    From Proposition \ref{prop:p1} we have that $R^G(T;\, \X_1) \subseteq [ \Phi^{E}( T;\, (\underline{x},\, \overline{x})) ],$ and $R^G(T;\, \X_1) = S^F(T; \X_1)$.
    Therefore $S^F(T;\, \X_1) \subseteq [\Phi^{E}( T;\, (\underline{x},\, \overline{x})) ]$.
\end{proof}

We next provide a special case for when the backward-time decomposition function is easily constructed from a (forward-time) decomposition function.

\begin{specialblock}
If
\begin{enumerate}
\item \eqref{eq1} is mixed-monotone with respect to $d$, and
\item for all $i\in\{1,\,\cdots,\, n\}$, we have $\frac{\partial d_i}{\partial x_i}(x,\, w,\, \widehat{x},\, \widehat{w}) \geq 0$ for  all $x,\,  \widehat{x} \in \X$ and for  all $w,\, \widehat{w} \in \W$ whenever the derivative exists,
\end{enumerate}
then \eqref{eq5} is mixed-monotone with decomposition function 
$D(x,\, w,\, \widehat{x},\, \widehat{w}) = -d (\widehat{x},\, \widehat{w},\, x,\, w)$.
\end{specialblock}

We demonstrate the bounding procedure from Proposition \ref{prop:p3} in the following example.

\begin{example}\label{Example3}
Consider the system
\begin{equation}\label{eq6}
\begin{bmatrix}
\dot{x}_1\\
\dot{x}_2
\end{bmatrix}
=
F(x,\, w)
=
\begin{bmatrix}
x_1 x_2 + w\\
x_1 + 1
\end{bmatrix}
\end{equation}
with $\X = \R^2$ and $\W = [0,\, 1/4]$.  The backward-time dynamics $\dot{x} = -F(x,\, w)$ for \eqref{eq6} are mixed-monotone with decomposition function 
\begin{equation}
\begin{split}
D_1(x,\, w,\, \widehat{x},\, \widehat{w}) & =
\begin{cases}
-x_1 \widehat{x}_2 - \widehat{w} & \text{if $x_1\geq 0$}, \\
-x_1 x_2 - \widehat{w} & \text{if $x_1< 0$},
\end{cases} 
\\
D_2(x,\, w,\, \widehat{x},\, \widehat{w}) & = -\widehat{x}_1 -1.
\end{split}
\end{equation}
Figure \ref{fig3} illustrates how finite-time backward reachable sets of \eqref{eq6} are approximated using Proposition \ref{prop:p3}.
\end{example}

\begin{figure}[t!]
    \hspace{.1in}
    \input{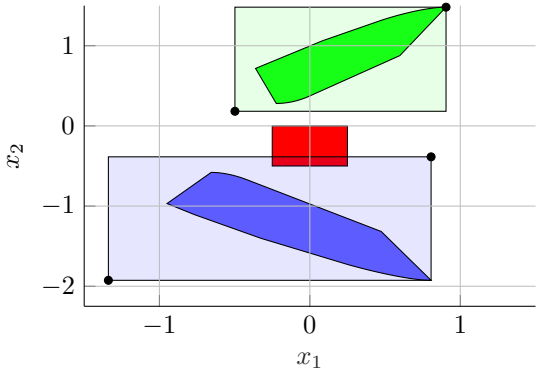}
    \caption{
	Approximating finite-time backward reachable sets of \eqref{eq6} with the set of initial conditions $\X_1 = [-1/4,\,1/4]\times[-1/2,\, 0]$. 
	$\X_1$ is shown in red.  
	$S^F(1;\, \X_1)$ is shown in blue, with a hyper-rectangular over approximation shown in light blue. 
	$R^F(1;\, \X_1)$ is shown in green, with a hyper-rectangular over approximation shown in light green.
	}
    \label{fig3}
\end{figure}

We next extend Theorem \ref{thrm1} to leverage the backward time dynamics \eqref{eq5}.  Specifically, we show that if \eqref{eq5} is mixed-monotone, as was the case in Proposition \ref{prop:p3}, then a robustly forward invariant region for \eqref{eq1} can be computed using an analogous technique to that of Theorem \ref{thrm1}.

\begin{theorem}\label{theorem2}
Let \eqref{eq5} be mixed-monotone with respect to $D$.
If there exists $(\underline{x},\, \overline{x}) \in \T$ such that $0  \preceq_{\rm SE} E(\underline{x},\, \overline{x})$ then $\X \setminus [\underline{x},\, \overline{x}]$ is robustly forward invariant for \eqref{eq1}.
\end{theorem}
\begin{proof}
If $0  \preceq_{\rm SE} E(\underline{x},\, \overline{x})$ holds for $(\underline{x},\, \overline{x}) \in \T$ then from Theorem \ref{thrm1} we have that $[\underline{x},\, \overline{x}]$ is robustly forward invariant for \eqref{eq5}.  Thus, $\X\setminus[\underline{x}, \overline{x}]$ is robustly forward invariant for \eqref{eq1}.
\qedhere
\end{proof}

\section{Case Study}
In this section, we present a numerical example to demonstrate the applicability of Theorems \ref{thrm1} and \ref{theorem2}.

Consider the system  
\begin{equation}\label{eq11}
\begin{bmatrix}
\dot{x}_1 \\
\dot{x}_2
\end{bmatrix}
= F(x,\, w) =
\begin{bmatrix}
-x_2 + x_1 (4 - 4x_1^2 - x_2^2) + w_1
\\
\:\:\; x_1 + x_2 (4 - x_1^2 - 4x_2^2) + w_2
\end{bmatrix}
\end{equation}
with $\X = \R^2$ and $\W = [-3/4,\, 3/4] \times [-3/4,\, 3/4]$.  The system \eqref{eq11} is mixed-monotone with respect to
\begin{equation*}
d(x,w,\widehat{x},\widehat{w}) = 
\begin{bmatrix}
-\widehat{x}_2 + 4x_1 - 4x_1^3 - l(x_1, x_2, \widehat{x}_2) + w_1
\\
\:\:\; x_1 + 4x_2 - 4x_2^3 - l(x_2, x_1, \widehat{x}_1) + w_2
\end{bmatrix}
\end{equation*}
\begin{equation*}
l(a,\, b,\, c) := 
\begin{cases}
a b^2 & 
\text{if } a \geq 0 \geq b \:\text{ and }\:    b \leq - c, \\
& \text{or } b \geq 0 > a \:\text{ and }\:  b \geq -c, \\
\vspace{-.3cm} \\
a c^2 
& \text{if } a,\, c \geq 0 \:\text{ and }\: b > -c, \\
& \text{or } a < 0 \:\text{ and }\: c \leq 0 \:\text{ and }\:  b < -c, \\
\vspace{-.3cm} \\
abc & 
\text{if } a \geq 0 \:\text{ and }\: b > 0 > c,\\
& \text{or } c > 0 > a,\,b .
\end{cases}
\end{equation*}
Additionally, we have $e(x_{\rm eq},\, \widehat{x}_{\rm eq}) = 0$ for 
\begin{equation}\label{eq12}
x_{\rm eq} = (-1.36,\, -1.36),\qquad
\widehat{x}_{\rm eq} = (1.36,\, 1.36).
\end{equation}
Therefore, from Theorem \ref{thrm1}, we have that $\X_{\rm eq} := [x_{\rm eq},\, \widehat{x}_{\rm eq}]$ is robustly forward invariant for \eqref{eq11}.
Additionally, $(x_{\rm eq},\, \widehat{x}_{\rm eq})$ is asymptotically stable on \eqref{eq:embedding} with a basin of attraction containing $\T$.  Therefore, $\X_{\rm eq}$ is globally attractive for \eqref{eq11}.

The backward-time dynamics $\dot{x} = -F(x,\, w)$ for \eqref{eq11} are mixed-monotone with decomposition function given by
\begin{multline*}
    D(x, w, \widehat{x}, \widehat{w}) = \\
\begin{bmatrix}
\:\:\; x_2 - 4x_1 + 4x_1^3 - l(-x_1, x_2, \widehat{x}_2) - \widehat{w}_1
\\
-\widehat{x}_1 - 4x_2 + 4x_2^3 - l(-x_2, x_1, \widehat{x}_1) - \widehat{w}_2
\end{bmatrix},
\end{multline*}
and we have $E(y_{\rm eq}, \widehat{y}_{\rm eq}) = 0$ for 
\begin{equation}\label{eq13}
y_{\rm eq} = (-0.59,\, -0.59),\qquad
\widehat{y}_{\rm eq} = (0.59,\, 0.59).
\end{equation}
Therefore, from Theorem \ref{theorem2}, we have that $\X \setminus \Y_{\rm eq}$ is robustly forward invariant for \eqref{eq11}, where  $\Y_{\rm eq} := [y_{\rm eq},\, \widehat{y}_{\rm eq}]$.
We show $\X_{\rm eq}$ and $\Y_{\rm eq}$ graphically in Figure \ref{fig4}.

\begin{figure}
    \hspace{.1in}
    \input{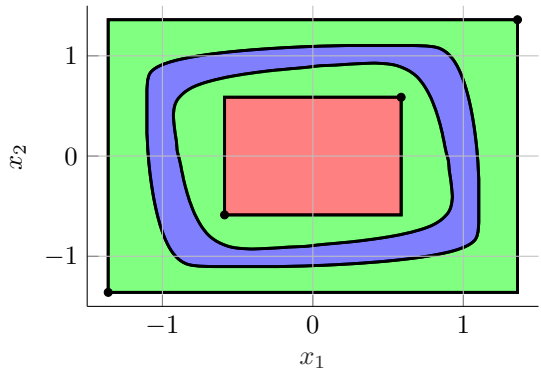}
    \caption{Computing robustly forward invariant sets for \eqref{eq11} by applying Theorems \ref{thrm1} and \ref{theorem2}.  $\X_{\rm eq}$ is the larger rectangle and $\Y_{\rm eq}$ is the smaller rectangle. The region shown in blue is the smallest attractive set computed numerically.}
    \label{fig4}
\end{figure}

\section{Conclusion}
This work presents several new reachability analysis tools for continuous-time mixed-monotone systems subject to a disturbance input.
The specific contributions of this paper are that 
(a) we suggest a new algorithm for computing decomposition functions for polynomial systems,
(b) we present an efficient method for explicitly computing robustly forward invariant sets for mixed-monotone systems, and
(c) we present a method for over-approximating finite-time backward reachable sets for mixed-monotone systems.

\bibliography{Bibliography.bib}
\bibliographystyle{ieeetr}
\end{document}